\documentclass[runningheads, envcountsame, a4paper]{llncs}

\usepackage[utf8]{inputenc}
\usepackage{tikz}
\usetikzlibrary{automata, arrows, chains, fit, shapes, decorations.pathreplacing}
\usepackage{amsmath}
\usepackage{amssymb}
\usepackage{xspace}
\usepackage{microtype}

\newcommand*{\N}{\mathbb{N}}

\newcommand*{\srel}{\mathcal{R}}
\newcommand*{\frel}{f}
\newcommand*{\perm}{\mathcal{P}}
\newcommand*{\permutations}[1]{\mathfrak{S}_{#1}}
\newcommand*{\colsep}{\#}

\newcommand*{\base}[3][]{\left\langle {#3} \right\rangle_{#1}}

\newcommand*{\compl}[1]{\overline{#1}}

\newcommand*{\size}[1]{\left| #1 \right|}

\newcommand*{\separator}{\$}

\newcommand*{\ie}{i.e.\xspace}
\newcommand*{\resp}{\emph{resp.}\xspace}

\newcommand*{\eg}{\emph{e.g.}\xspace}

\newcommand*{\ssi}{\emph{iff}\xspace}

\newcommand*{\DFA}{\mathrm{DFA}}
\newcommand*{\NFA}{\mathrm{NFA}}
\newcommand*{\AFA}{\mathrm{AFA}}
\newcommand*{\coAFA}{\mathrm{co}$-$\mathrm{AFA}}

\newcommand*{\bmark}{\vdash}
\newcommand*{\emark}{\dashv}
\newcommand*{\bomark}{\#}

\newcommand*{\alphabet}{\Sigma}
\newcommand*{\atm}{\mathcal{A}}

\newcommand*{\states}{Q}
\newcommand*{\init}{I}
\newcommand*{\final}{F}
\newcommand*{\transrel}{\Delta}
\newcommand*{\goleft}{\vartriangleleft}
\newcommand*{\goright}{\vartriangleright}
\newcommand*{\stay}{\triangledown}
\newcommand*{\moves}{\{\goleft, \stay, \goright\}}
\newcommand*{\markers}{\{\bmark, \emark\}}
\newcommand*{\addbounds}[1]{{#1}^\emark_\bmark}
\newcommand*{\ealph}{\addbounds{\alphabet}}
\newcommand*{\alphb}{\alphabet_\bmark}

\newcommand*{\some}{a}
\newcommand*{\eword}[1]{{#1}^\emark_\bmark}

\newcommand*{\reverse}[1]{\tilde{#1}}
\newcommand*{\model}{two-way two-tape automaton\xspace}

\renewcommand*{\models}{two-way two-tape automata\xspace}
\newcommand*{\MODELS}{Two-way Two-tape Automata\xspace}

\newcommand*{\dmodel}{deterministic two-way two-tape automaton\xspace}

\newcommand*{\omodel}{two-way automaton\xspace}

\newcommand*{\amodel}{alternating two-way two-tape automaton\xspace}
\newcommand*{\amodels}{alternating two-way two-tape automata\xspace}

\newcommand*{\nfa}{non-deterministic finite automaton\xspace}
\newcommand*{\wordsof}[1]{{#1}^*}
\newcommand*{\words}{\wordsof{\alphabet}}
\newcommand*{\picturesof}[1]{{#1}^{**}}

\newcommand*{\upictures}{\picturesof{\{\some\}}}

\newcommand*{\tensorprod}[2]{#1 \otimes #2}

\newcommand*{\enum}{\mathcal{E}}
\newcommand*{\incrementS}{\mathcal{I}}
\newcommand*{\increment}{\mathcal{J}}

\newcommand*{\adaptarrow}[2][]{\mathchoice
         {\xrightarrow{#2}_{#1}}
         {\xrightarrow{\smash{\lower1pt\hbox{$\scriptstyle #2$}}}_{#1}}
         {\text{Erreur}}
         {\text{Erreur}}}
\newcommand*{\xadaptarrow}[2][]{\mathchoice
         {\xrightarrow[#1]{#2}}
         {\xrightarrow[\smash{\lower1pt\hbox{$\scriptstyle #1$}}]{\smash{\lower1pt\hbox{$\scriptstyle #2$}}}}
         {\text{Erreur}}
         {\text{Erreur}}}

\newcommand*{\instruction}[4]{{#1}, {#2} \mid {#3}, {#4}}
\newcommand*{\transition}[7][]{#2 \adaptarrow[#1]{\instruction{#3}{#4}{#5}{#6}} #7}

\newcommand*{\stransition}[5][]{#2 \adaptarrow[#1]{#3\mid #4} #5}

\newcommand*{\xleadsto}[1]{\overset{#1}{\leadsto}}
\newcommand*{\xadaptleadsto}[1]{\mathchoice
         {\xleadsto{#1}}
         {\xleadsto{\smash{\lower1pt\hbox{$\scriptstyle #1$}}}}
         {\text{Erreur}}
         {\text{Erreur}}}

\newcommand*{\succconf}[3][]{#2 \xadaptleadsto{#1} #3}

\newcommand*{\rela}[1]{\mathcal{R}\left(#1\right)}

\newcommand*{\Coprime}{\mathrm{Coprime}}

\begin{document}
\frontmatter          
\pagestyle{headings}  
\addtocmark{Two-way Two-tape Automata} 
%
%
\mainmatter              
\title{Two-way Two-tape Automata}
\toctitle{Two-way Two-tape Automata}
%
%
\author{Olivier Carton\inst{1}\fnmsep\thanks{funded by the DeLTA project (ANR-16-CE40-0007)} \and
      L\'{e}o Exibard\inst{2} \and Olivier Serre\inst{1}}
\authorrunning{O. Carton, L. Exibard and O. Serre} 
\tocauthor{Olivier~Carton, L\'{e}o~Exibard and Olivier~Serre}
%
%
\institute{IRIF,
  Universit\'{e} Paris Diderot \& CNRS
\and
D\'{e}partement d'Informatique,
ENS de Lyon
}

\maketitle              

\setcounter{footnote}{0} 

\begin{abstract}
  In this article we consider two-way two-tape (alternating) automata
  accepting pairs of words and we study some closure properties of this
  model.  Our main result is that such alternating automata are not closed
  under complementation for non-unary alphabets.  This improves a similar
  result of Kari and Moore for picture languages.  We also show that these
  deterministic, non-deterministic and alternating automata are not closed
  under composition.
\end{abstract}

\section{Introduction}

In this article we consider two-way two-tape (alternating) automata that
are designed to recognize binary relations between finite words. Although
these automata are quite natural as read-only Turing machines, almost no
work has been devoted to this model of computation.  We study their
properties, with a special focus on their closure properties, in particular
under complementation and composition.

Finite states machines with inputs and outputs are widely used in many
different areas like coding \cite{LindMarcus95}, computer arithmetics
\cite{Frougny02}, natural language processing \cite{RocheSchabes97} and
program analysis \cite{CohenCollard98}.  The simplest model is obtained by
adding outputs to a classical finite-state (non)-deterministic one-way
automaton to get a machine known as a transducer. In a transducer, the
input word is only scanned once by a one-way head and the output is
produced by the transitions used along the reading. A transducer can
equivalently be seen as a machine with two tapes, one for the input and the
other for the output, that are scanned once by two one-way heads.
Relations realized by this kind of machines are called rational.  They have
been intensively studied since the early days of automata theory
\cite{ElgotMezei65} and they enjoy some nice properties
\cite{Sakarovitch09}.  They are, for instance, closed under composition,
but not under complementation.

Rational relations turn out to form a rather small class, hence classes of
stronger transducers have been introduced by enriching transducers with
extra features like two-wayness and/or alternation. A well studied class is
that of two-way transducers in which the input word is scanned by a two-way
head and the output is produced (or equivalently scanned) by a one-way
head.  In that class, deterministic machines are of special interest as
they are equivalent to MSO-transductions and turn out to be closed under
composition \cite{EngelfrietHoogeboom01}.

In this article, we consider machines with two tapes which are both scanned
by two-way heads.  In this model, it is important to consider the output
word as already written, and not produced on some output tape to ensure
consistency, as it is scanned several times.

Another key feature, used to increase either the expressive power or the
succinctness of finite state machines, is alternation, which allows the
machine to spawn several copies of itself.  Alternation often provides for
free closure under complementation for machines on finite structures like
words or trees. Indeed, the dual machine obtained by swapping existential
and universal states and complementing the acceptance condition accepts the
complement language as long as all computations terminate: if the run
of a machine may loop, the closure under complementation of alternating
machines may no longer hold.

Picture automata introduced in \cite{BlumHewitt67} scan a 2-dimensional
array of symbols with moves in the four cardinal directions to either
accept or reject it.  As in the case of two-way automata, the border of the
array is marked by special symbols. A run of such an automaton may loop as
it can scan the same position twice with the same control state.  These
picture automata differ from classical word or tree automata where all
variants (deterministic, non-deterministic, alternating) are equivalent.
Indeed, it has been shown by Kari and Moore that alternating picture
automata are not closed under complementation as soon as the alphabet size
is greater than~1 \cite{KariMoore01}.  This means that loops are inherent
to the model and that they cannot be removed.

In this article, we show that two-tape two-way alternating automata are not
closed under complementation either.  Picture automata are actually very
close to the model that we consider.  In particular they coincide for unary
alphabets.  Indeed, over a unary alphabet, a pair of words is merely a pair
of integers (their lengths) and this is equivalent to a two-dimensional
array on a unary alphabet.  However, as soon as the alphabets have
cardinality at least $2$, the models are distinct.  Indeed, for an alphabet
of size~$k$, the number of $m \times n$-arrays is $k^{mn}$ while the number
of pairs of words is only $k^{m+n}$.

The fact that the two models coincide for unary alphabets allows us to
recover immediately some separation and undecidability properties.  For
instance, it has been shown that such deterministic picture automata are
strictly less powerful than non-deterministic ones which are, in turn, less
powerful than alternating automata \cite{KariMoore01}. These results carry
over the two-tape two-way automata that we consider.

To prove that alternating automata are not closed under complementation, we
use a counter-example that is close to the one in~\cite{KariMoore04} for
picture automata. However, since our coding is different, we need some
extra arguments to show that it is accepted by an alternating automaton
and, to show that its complement cannot, we use a more direct proof.

\section{Two-way Two-tape Automata}

In this article, $\alphabet$ is a finite alphabet and $\words$ denotes the
set of finite words over $\alphabet$. For a word $u \in \words$, we denote
its length by $\size{u}$, and for each $1\leq i\leq \size{u}$ we denote by
$u_i$ its $i$-th letter. From now on, $\bmark$ (begin) and $\emark$ (end)
are reserved characters not belonging to $\alphabet$ and marking word
boundaries. For simplicity of notation, we let $\ealph = \alphabet \cup
\markers$.  For $u \in \words$, we let $\eword{u} = {\bmark}u{\emark}$,
with ${\eword{u}}_0 = {\bmark}$, ${\eword{u}}_{\size{u}+1} = {\emark}$ and
${\eword{u}}_i = u_i$ for each $1 \leq i \leq \size{u}$.

A \emph{(non-deterministic) two-way two-tape finite automaton} is a tuple
$\atm = \left(\states, \alphabet, \transrel,\init, \final \right)$, where
$\states$ is the set of \emph{states}, and $\init, \final \subseteq
\states$ are respectively the sets of \emph{initial} and \emph{final}
states.  We call $\transrel \subseteq \left(\states \times {\ealph} \times
  {\ealph}\right) \times (\states \times \moves \times \moves)$ the
\emph{transition relation}. We use the notation $\stransition{p}{a_1,
  a_2}{d_1, d_2}{q}$ for $(p, a_1, a_2, q, d_1, d_2) \in \transrel$. We
require that the reading heads cannot cross the words boundaries, \ie for
every transition $\stransition{p}{a_1, a_2}{d_1, d_2}{q}$ and every $i=1,2$
if $a_i = {\bmark}$ (\resp $a_i = {\emark}$) then $d_i \neq {\goleft}$
(\resp $d_i \neq {\goright}$).

An automaton~$\atm$ is said to be \emph{deterministic} whenever for every
state $p \in \states$ and every letters $ a_1, a_2 \in \alphabet$, there
exists at most one $q \in \states, d_1, d_2 \in \moves$ such that
$\stransition{p}{a_1, a_2}{d_1, d_2}{q}$.

Note that extending the model to more than two tapes is straightforward.
Note also that if we restrict the model to a single tape we retrieve the
classical notion of two-way automata on finite words.

\begin{figure}[ht]
  \centering
  \begin{tikzpicture}

  \edef\sizetape{0.5cm}
  \tikzstyle{tmtape}=[draw,minimum size=\sizetape]

  \begin{scope}
    \node[draw, rounded corners, minimum height=0.7cm, minimum width=0.7cm] (Q) {$Q$};
  \end{scope}

  \begin{scope}[shift={(2cm,0.8cm)}, start chain=1 going right,node distance=-0.15mm]
      \node [on chain=1,tmtape] {$\bmark$};
      \node [on chain=1,tmtape] {$u_1$};
      \node [on chain=1,tmtape, minimum width=1cm] {$\cdots$};
      \node [on chain=1,tmtape, thick] (input1) {$u_i$};
      \node [on chain=1,tmtape, minimum width=1.5cm] {$\cdots$};
      \node [on chain=1,tmtape] {$u_m$};
      \node [on chain=1,tmtape] (end1) {$\emark$};
  \end{scope}
      \node[right=of end1] (legend1) {First Tape};

  \begin{scope}[shift={(2cm,-0.8cm)}, start chain=1 going right,node distance=-0.15mm]
      \node [on chain=1,tmtape] {$\bmark$};
      \node [on chain=1,tmtape] {$v_1$};
      \node [on chain=1,tmtape, minimum width=1.5cm] {$\cdots$};
      \node [on chain=1,tmtape, thick] (input2) {$v_j$};
      \node [on chain=1,tmtape, minimum width=1.5cm] {$\cdots$};
      \node [on chain=1,tmtape] {$v_n$};
      \node [on chain=1,tmtape] (end2) {$\emark$};
  \end{scope}
      \node (legend2) at (legend1 |- end2) {Second Tape};

    \draw[->, rounded corners=7pt] (Q.north east) ++ (0,-0.25cm) -| (input1);
    \draw[->, rounded corners=7pt] (Q.south east) ++ (0, 0.25cm) -| (input2);

    \draw[->] (input1.south) ++ (-0.25cm, -0.25cm) -- ++(-0.5cm, 0);
    \draw[->] (input1.south) ++ (0.25cm, -0.25cm) -- ++(0.5cm, 0);

    \draw[->] (input2.north) ++ (-0.25cm, 0.25cm) -- ++(-0.5cm, 0);
    \draw[->] (input2.north) ++ (0.25cm, 0.25cm) -- ++(0.5cm, 0);
  \end{tikzpicture}
  \caption{Schema of a two-way two-tape finite automaton}
  \label{schemaTwttfa}
\end{figure}
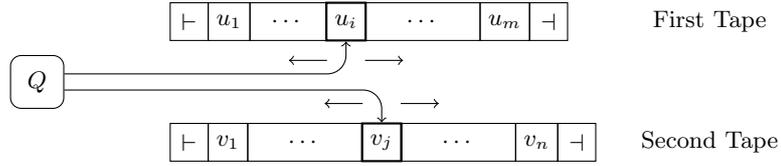

A \emph{configuration} is a triple $(q, i, j) \in \states \times \N \times
\N$, where $q$ is the current state, and $i$ (resp. $j$) is the position of
the reading head on the first (resp. second) tape (recall that the $\bmark$
marker is by convention at position $0$).

We say that a configuration $(q_2, i_2, j_2)$ is a \emph{successor} of a
configuration $(q_1, i_1, j_1)$ with regard to input $(u, v)$, written
$\succconf[(u,v)]{(q_1, i_1, j_1)}{(q_2, i_2, j_2)}$, when automaton
$\atm$ can go from one configuration to the next in a single step, \ie if
$\transition{q_1}{(\eword{u})_{i_1}}{(\eword{v})_{j_1}}{d}{e}{q_2}$,
and if $i_2 = i_1 + \chi(d)$ and $j_2 = j_1 + \chi(e)$, where
$\chi(\goleft) = -1$, $\chi(\stay) = 0$ and $\chi(\goright) = 1$.

A \emph{run} of $\atm$ on input $(u, v) \in \words \times \words$ is a
(possibly infinite) sequence $(p_k, i_k, j_k)_{1 \leq k < n}$, where $n \in
\N \cup \{\infty\}$, of successive configurations: for every $1 \leq k <
n$, one has $\succconf[(u,v)]{(p_k, i_k, j_k)}{(p_{k+1}, i_{k+1},
  j_{k+1})}$.

An \emph{initial run} is a run starting with an initial configuration, \ie
one such that $p_0 \in \init$ and $i_0 = j_0 = 0$.  It is \emph{accepting}
when it is moreover finite and contains a final state $f \in \final$, \ie
there exists $k \leq n$ such that $ p_k \in \final$.

The \emph{relation $\rela{\atm}$ accepted by} a \model $\atm$ is the set of
pairs $(u,v)$ such that there exists an accepting run of~$\atm$ on input
$(u,v)$.

We now introduce alternating automata, which generalize non-deterministic
automata. An \emph{alternating automaton} is an automaton whose set of
control states is partitioned into \emph{existential} ($Q_\exists$) and
\emph{universal} ($Q_\forall$) states. A configuration is defined as in the
non-deterministic setting and it is existential (\resp universal) if the
control state is.

\emph{Runs} of alternating automata are (possibly infinite) trees whose
nodes are labeled by configurations, and such that each inner node $u$
labeled by a configuration~$C$ satisfies the following conditions:
\begin{itemize} 
\item If $C$ is existential then $u$ has a single son that is
  labeled by a successor configuration of $C$.
\item If $C$ is universal and if $\{C_1,\dots,C_k\}$ denotes all successor
  configurations of $C$, then $u$ has $k$ sons each of them labeled by a
  different $C_i$ for $1\leq i\leq k$.
\end{itemize}

A run is \emph{accepting} if it is finite, its root is labeled by an
initial configuration and all its leaves are labeled either by accepting
configurations, or by universal configurations that have no successor
configuration. Again, we define the relation accepted by an alternating
automaton as those pairs of words over which there is an accepting run.
Note that non-deterministic automata correspond to the case where
$Q_\forall = \emptyset$.

First we remark that if we restrict the model by forbidding the reading
heads to go to the left, we obtain a $1$-way model that recognizes the
rational relations, \ie those realized by transducers \cite{Sakarovitch09}.
Not surprisingly this is a restriction as \models can, for instance,
recognize deterministically the relation $\{(u, \reverse{u}) \mid u \in
\words\}$, where $\reverse{u}$ denotes the reverse of $u$.  The
corresponding automaton is depicted in Figure~\ref{reverse}.  This relation
cannot be recognized by classical one-way transducers.

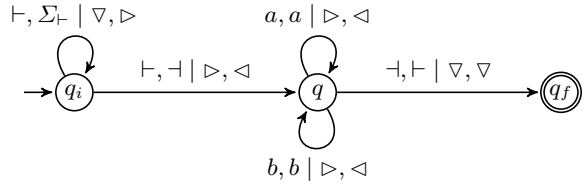
\begin{figure}[ht]
  \centering
  \begin{tikzpicture}[->,>=stealth',shorten >=1pt,auto,node distance=3.2cm, semithick]
    \tikzstyle{every state}=[inner sep=0em, minimum size=1.5em]

  \node[initial, initial text =, state] (A)                    {$q_i$};
  \node[state, right of = A]            (B)                    {$q$};
  \node[accepting, state, right of= B]  (C)                    {$q_f$};

  \path (A) edge [out=120,in=60,loop] node {$\instruction{\bmark}{\alphb}{\stay}{\goright}$} (A);
  \path (A) edge              node {$\instruction{\bmark}{\emark}{\goright}{\goleft}$} (B);
  \path (B) edge [out=120,in=60,loop] node {$\instruction{a}{a}{\goright}{\goleft}$} (B);
  \path (B) edge [out=300,in=240,loop] node {$\instruction{b}{b}{\goright}{\goleft}$} (B);
  \path (B) edge              node {$\instruction{\emark}{\bmark}{\stay}{\stay}$} (C);
  \end{tikzpicture}
  \caption{Automaton recognizing $\{(u, \reverse{u}) \mid u \in \words\}$ when $\Sigma = \{a,b\}$}
  \label{reverse}
\end{figure}
  
Our model captures much more complex relations. It is shown in the next
section that \models can recognize the relation $\Coprime = \{(a^p, a^q)
\mid p \wedge q = 1\}$, where $p \wedge q$ denotes the greatest common
divisor of $p$ and $q$, by implementing a variant of the Euclidean
algorithm.

\section{Picture Languages} \label{4wayAutomata}

For brevity we simply recall here some key notions on picture languages: we
refer the reader to \cite{GiammarresiRestivo92} for an excellent survey on
the topic with complete definitions of all objects discussed below.

A \emph{picture} $p$ of dimensions $m \times n$ is a matrix over a finite
alphabet $\alphabet$. For every $1\leq i\leq m$ and $1\leq j\leq n$, we
write $p_{i,j}$ for the content of the cell at position $(i,j)$. To
recognize pictures, we add a special marker $\bomark\notin\alphabet$ all
around the picture $p$, \ie we adopt the convention that for all $0 \leq i
\leq m+1$ and $0 \leq j \leq n+1$, $p_{0,j} = p_{m+1,j} = p_{i,0} = p_{i,
  n+1} = \bomark$.  We write $\picturesof{\alphabet}$ for the set of all
pictures over $\alphabet$. A \emph{picture language} is thus a subset of
$\picturesof{\alphabet}$.

In order to recognize picture languages, 4-way automata were first
introduced in~\cite{BlumHewitt67}. Such an automaton has a single head
which is able to move on a two-dimensional array of symbols (surrounded by
markers) in the four directions (up, down, left, right) and accepts when
reaching a final state. A schema is provided in Figure~\ref{schema4wfa}.

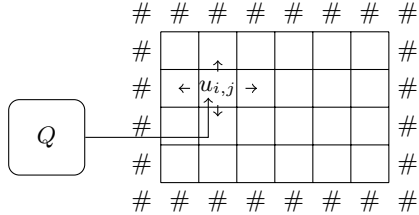
\begin{figure}[ht]
  \centering
  \begin{tikzpicture}[scale=0.5]

  \edef\sizetape{0.8cm}
  \tikzstyle{tmtape}=[draw,minimum size=\sizetape]

  \begin{scope}
    \node[draw, rounded corners, minimum height=1cm, minimum width=1cm] (Q) {$Q$};
  \end{scope}

  \begin{scope}[shift={(3cm,-1.2cm)}]
    \draw[step=1.0, black] (0,0) grid (6,4);
    \foreach \i in {-0.5,6.5}
    \foreach \j in {-0.5, ..., 4.5}
    {\node at (\i, \j) {\bomark};}
    \foreach \j in {-0.5,4.5}
    \foreach \i in {0.5, ..., 5.5}
    {\node at (\i, \j) {\#};}
    \node (input) at (1.5, 2.5) {$u_{i,j}$};
    \coordinate (inputb) at (1.25, 2.25);
  \end{scope}

    \draw[->] (Q.east) -| (inputb);
    \draw[->] (input.east) -- ++(0.3cm,0);
    \draw[->] (input.west) -- ++(-0.3cm,0);
    \draw[->] (input.north) -- ++(0,0.3cm);
    \draw[->] (input.south) -- ++(0,-0.3cm);

  \end{tikzpicture}
  \caption{Schema of a 4-way automaton}
  \label{schema4wfa}
\end{figure}

The link between \models and picture automata is provided by special
pictures, called \emph{products of words}, that we now define. For $u, v
\in \words$, we define the picture $\tensorprod{u}{v} = ((u_i,
v_j))_{\tiny{\begin{array}{l} 1 \leq i \leq \size{u} \\ 1 \leq j \leq
      \size{v} \end{array}}}$ over the product alphabet $\alphabet \times \alphabet$.

Thus, any relation $\srel \subseteq \words \times \words$ is mapped to a
picture language $L_{\srel}^{\otimes} \subseteq \picturesof{(\alphabet
  \times \alphabet)}$.  Moreover, over unary alphabets, pictures languages
and binary relations over words are in one-to-one correspondance as the
pair $(a^m, a^n)$ unambiguously represents an image of dimensions $m \times
n$ and conversely.  Consequently, for $L \subseteq \upictures$,
$L$ is recognizable by a $4$-way deterministic (resp. non-deterministic,
alternating) automaton \ssi $\srel$ is recognizable by a deterministic
(resp. non-deterministic, alternating) \model where $\srel$ is the unique
relation such that $L = L_{\srel}^{\otimes}$.

Thus, all the results known for unary picture languages also hold in our
model when $\size{\alphabet} = 1$. In particular, in~\cite{KariMoore01}, it
is shown that determinism is strictly weaker than non-determinism, the
latter being weaker than alternation: $\DFA \subsetneq \NFA \subsetneq
\AFA$ where $\DFA$ (resp. $\NFA$, $\AFA$) is the class of relations
recognized by deterministic (resp. non-deterministic, alternating) \models.
The relation $\srel = \{(a^w, a^h) \mid \exists i,j \in \N, w = ih +
j(h+1)\}$ is such that $\srel \notin \DFA, \srel \in \NFA, \compl{\srel}
\notin \NFA, \compl{\srel} \in \AFA$, where $\compl{\srel}$ denotes the
complement of~$\srel$ in $ \words \times \words$. In \cite{KariSalo11}, it
is shown that the emptiness problem is undecidable even for a unary
alphabet.

\begin{figure}[ht]
  \centering
  \begin{tikzpicture}[scale=0.6, baseline=(current bounding box.west)]
    \tikzstyle{reddot} = [draw, shape=circle, radius=1pt, inner sep=1pt, fill=red]
  \draw[dashed] (0,0) grid (7,3);
  \draw (0,0) rectangle (7,3);
  \draw (0,0) -- (3,3) -- (6,0) -- (7,1);
  \node[reddot, label=right:{$(p,1)$}] (gcd1a) at (7,1) {};
  \node[reddot, label=above:{$(1,q)$}] (gcd1b) at (1,3) {};
  \draw[decorate, decoration={brace, amplitude = 5pt}] (7,0) -- (0,0) node[midway, below=6pt] (p) {$p$};
  \draw[decorate, decoration={brace, amplitude = 5pt}] (0,0) -- (0,3) node[midway, left=6pt] (q) {$q$};
  \end{tikzpicture}
  \quad
  \begin{tikzpicture}[scale=0.3, baseline = (current bounding box.west)]
  \draw[dashed] (0,0) grid (16, 16);
  \draw (0,0) rectangle (16,16);
  \draw (0,0) -- (16,8) -- (12,0) -- (16,2) -- (15,0);
  \draw[very thick] (15,0) rectangle (16,1);
  \draw[decorate, decoration={brace, amplitude = 5pt}] (16,0) -- (0,0) node[midway, below=6pt] (L) {$2^n$};
  \draw[decorate, decoration={brace, amplitude = 5pt}] (0,0) -- (0,16) node[midway, left=6pt] (l) {$2^n$};
  \end{tikzpicture}
  \caption{Euclidean algorithm (left) and Squares of size $2^n$ (right)}
  \label{coprime}
  \label{expsquare}
\end{figure}
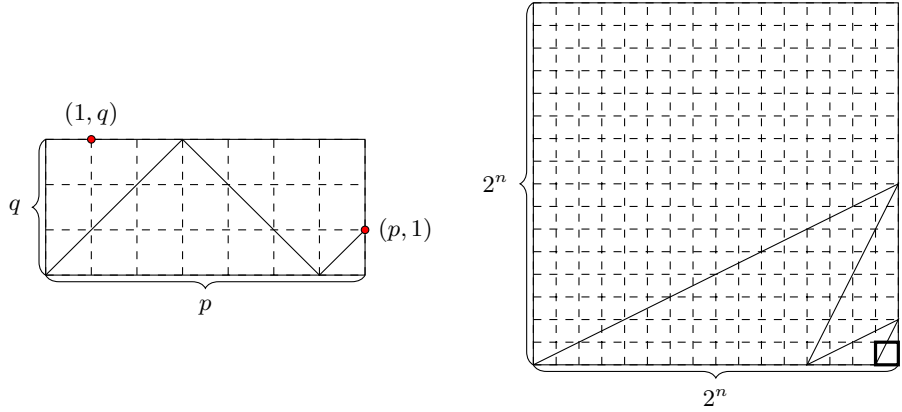

The following examples from~\cite{KariMoore04} show that picture automata
and \models are really expressive.  The relation $\Coprime = \{(a^p, a^q)
\mid p \wedge q = 1\}$ can be recognized by a deterministic automaton
implementing a variant of the Euclidean algorithm.  The automaton follows
diagonals until it reaches either one of $(p,1)$ or $(1,q)$ or one of
$(0,0)$, $(p,0)$, $(0,q)$ or $(p,q)$.  In the former case, it accepts and
in the latter case, it rejects.

A more sophisticated deterministic automaton can recognize the relation
$\{(2^n, 2^n) \mid n \in \N\}$ following the schema in
Figure~\ref{expsquare} (right).  The automaton moves with a ratio of $1/2$
as long as it is possible (\ie while the remaining length is divisible by
2), and it accepts if it reaches the bottom-right square.  An even more
sophisticated deterministic automaton can accept the relation $\{(2^{2^n},
2^{2^n}) \mid n \in \N\}$.

\section{Alternating \MODELS Are Not Closed under Complementation}
\label{proofAFAcompl}

Our main result is the following. Note that the case of unary alphabets is
still an open problem.

\begin{theorem} \label{AFAcompl}
  The class of relations recognized by \amodels is not closed under
  complementation as soon as the alphabet is not unary.
\end{theorem}

In this proof we denote by $\AFA$ the class of relations recognized by
\amodels, and by $\coAFA$ the class of relations whose complement is
recognized by \amodels. Hence, we aim to prove that $\AFA$ and~$\coAFA$
are distinct, and this is achieved by defining a well-chosen relation
$\perm$ and show that $\perm\in \AFA$ but $\perm\notin \coAFA$. The first
point is proved by giving an explicit \amodel recognizing $\perm$
(Lemma~\ref{PisinAFA}); the second point is proved by contradiction using a
game-theoretic approach combined with a combinatorial argument
(Lemma~\ref{PnotincoAFA}).

The proof is inspired by the one in \cite{KariMoore01} establishing that
the set of picture languages recognized by alternating 4-way automata is
not closed under complementation. One difference is that the
counter-example they exhibit cannot be used in our setting as the images
that are considered are not product of words (as defined in
Section~\ref{4wayAutomata}). Hence, even if the general idea ---~coding a
permutation to build a counter-example~--- is similar to the one in
\cite{KariMoore01}, our coding is different and therefore new ideas are
needed to prove that we can accept this relation with an \amodel. The proof
that the complement is not recognizable by an \amodel, even if it shares
ideas with the one in \cite{KariMoore01}, is somehow more direct as it does
not appeal to an intermediate class.

In the following, if $c_1, \dots, c_n$ are $n$ words of length $n$ over
alphabet $\{0,1\}$ we identify the tuple $(c_1,\dots,c_n)$ with the
$n\times n$ matrix whose $i$-th column is $c_i$. We define the relation
$\perm$ by $$\perm = \left\{(\some^n, c_1 \colsep \cdots \colsep c_n
  \separator c_1 \colsep \cdots \colsep c_n) \mid n \in \N, (c_1, \dots,
  c_n) \in \permutations{n} \right\}$$
where $\permutations{n}$ denotes the
set of all ($n\times n$ $\{0,1\}$-matrix coding) permutations over $\{1,
\dots, n\}$. See Figure~\ref{identicalPermutations} for an example.

\begin{figure}[ht]
    \centering
$\def\arraystretch{1.5}
\begin{array}{|ccccccc|}
   \hline
   0 & 0 & 1 & \separator & 0 & 0 & 1 \\
   1 & 0 & 0 & \separator & 1 & 0 & 0 \\
   0 & 1 & 0 & \separator & 0 & 1 & 0 \\
   \hline
\end{array}$ \hspace{3em}
$(\some^3, 010\colsep 001 \colsep 100 \separator 010\colsep 001 \colsep 100)$
\caption{Two identical permutations separated by $\separator$s and the corresponding encoding}
    \label{identicalPermutations}
  \end{figure}%

As announced, we start by showing that $\perm \in \AFA$.

\begin{lemma}
  \label{PisinAFA}
There exists an \amodel recognizing $\perm$.
\end{lemma}

\begin{proof}
  We describe below how to check whether a pair $(\some^n, c) \in
  \wordsof{\some} \times \wordsof{\{0,1,\separator,\colsep\}}$ is in $\perm$. As the
  $\some^n$ word is used only to store position we sometimes refer to
  the content of the first tape as the \emph{counter}.
  
  To decide if $(a^n,c)\in\perm$, we have to check two things. The first
  one is whether $c$ is of the form $\sigma_1 \separator \sigma_2$, where $\sigma_1,
  \sigma_2$ are the encodings of two permutations of dimension $n$ and 
  the second one is whether $\sigma_1 = \sigma_2$.
  
  For the first step, we only explain how to check the property for
  $\sigma_1$ as the case of $\sigma_2$ is checked in the same way. Assume
  $\sigma_1 = c_{0} \colsep c_{1} \colsep \cdots \colsep c_{n}$ (checking
  that there are $n$ block is easy). We first need to check that each $c_i$
  (\ie each column) has length $n$ and contains exactly one $1$: this is
  easy (the length condition being checked thanks to the first tape). We
  then need to check that for each $k=1,\ldots,n$ there is exactly one
  $c_i$ whose $k$-th letter is a $1$ (\ie each row contains exactly one
  $1$). This property is checked for each $k=1,\dots,n$ in increasing
  order starting from $k=1$, and at the beginning of step $k$ the head in
  the counter in the first tape is at position $k$ and the head in the
  second tape is at the beginning of $\sigma_1$: now going left on the
  counter and right on the second tape at the same speed the $k$-th symbol
  of $c_1$ is reached (just before the left marker is read on the first
  tape); then the automaton goes back to the beginning of $c_1$ and the
  counter is increased back to $k$; finally the second head goes to the
  beginning of $c_2$ (going to the right until reading a $\#$), then $c_2$
  is processed in the same way, and so on until reading a $\separator$ meaning
  that $c_n$ was processed and that one can go back to the beginning of the
  second tape and start again but now for $k+1$.
  
  We are now left with checking that $\sigma_1=\sigma_2$ knowing that both
  $\sigma_1$ and $\sigma_2$ encode a permutation. For that we use the same
  approach as the one in \cite[Lemma~2, Condition~(*)]{KariMoore01}. More
  precisely, we use the following property: two permutations are different
  \ssi there exists an inversion \ie there exists $i \neq i', \sigma_1(i) <
  \sigma_1(i')$ but $\sigma_2(i) > \sigma_2(i')$. This can be checked by
  trying all possible way of moving in the associated picture (as
  illustrated in Figure~\ref{identicalPermutations}) in the following
  fashion (see Figure~\ref{inversion}): from a $1$ move right to another
  column but stay on the same side of the column of $\separator$s. Find the $1$ on
  that column. Then move to the other 1 that is on the same row, on the
  opposite side of the column of $\separator$s, and repeat: then, the machine enters
  an infinite loop \ssi it finds an inversion~\cite{KariMoore01}.

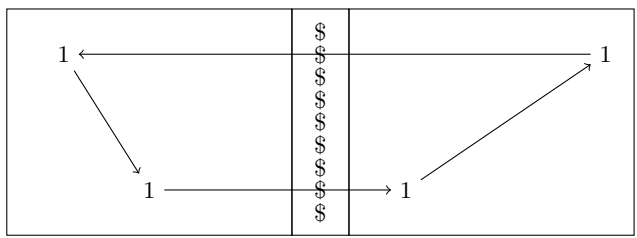
\begin{figure}[ht]
 \centering
\begin{tikzpicture}[xscale=0.75, yscale=0.6]
  \draw (0,0) rectangle (5,5);       
  \draw (5,5) rectangle (6,0);       
  \foreach \j in {0.5, 1, ..., 4.5}  
  \node at (5.5, \j) {$\separator$};
  \draw (6,0) rectangle (11,5);      
  
  \node (i1)  at (1, 4) {$1$};
  \node (i'1) at (2.5, 1) {$1$};
  \node (i2)  at (10.5, 4) {$1$};
  \node (i'2) at (7, 1) {$1$};
  
  \draw[->] (i1) edge (i'1) (i'1) edge (i'2) (i'2) edge (i2) (i2) edge (i1);

\end{tikzpicture}
\caption{Loop corresponding to an inversion}
  \label{inversion}
\end{figure}

The first step (finding a $1$) is easy and done using universal states
(which ensures that we do the check for all possible $1$). Moving to
another column on the right is easy: simply use a universal state and keep
moving the head to the right passing one or more $\#$ but stopping before
reading the $\separator$. Finding the $1$ on the column is easy (look for it in
between the left and right $\#$). Finally, moving to the other 1 on the
same row on the opposite side of the $\separator$ is achieved by using the counter
to keep in memory the row number $j$, and then go to the other side and
check every column until finding the one with a $1$ on the $j$-th row (this
is done with the same trick used previously to check that each row contains
a single $1$). 
If at some point when moving to a column to the right the automaton hit a $\separator$ then it goes to a final state. The automaton accepts the desired language for the following reason: if an inversion exists, it is found (thanks to universal choices) and leads to a looping hence, rejecting, computation; if not, whatever the universal choices are the automaton is getting closer and closer to the $\separator$ and eventually hit it thus  reaching a final state.

Therefore, we built an \amodel recognizing $\perm$. More generally, the ideas used in the first part of the proof can be used to prove the following: for any picture language $L$ recognizable by a 4-way automaton, the relation $\{(a^{\max(n,m)},c_1\#\cdots\#c_m)\mid (c_1,\dots,c_m) \text{ is an $n\times m$ picture in $L$}\}$ is recognizable by a \model.
\qed\end{proof}

We now show that the complement of $\perm$ cannot be recognized by an \amodel.

\begin{lemma}\label{PnotincoAFA}
  Relation $\perm\notin\coAFA$.
\end{lemma}

\begin{proof}
  The proof is by contradiction assuming the existence of an \amodel
  $\atm=\left(\states, \alphabet, \transrel,\init, \final \right)$ such
  that a pair $(u,v)$ is accepted by $\atm$ \ssi $(u,v)\notin\perm$. Hence,
  a pair $(a^n,c)$ is rejected by $\atm$ \ssi if belongs to $\perm$.
        
  We now rephrase the fact that $(a^n,c)$ is rejected by $\atm$, to that end 
  we use a game-theoretic flavoured argument. Indeed, another way
  of thinking about a run of an alternating automaton is as a two-player games
  where the existential player is in charge of choosing the transition in
  existential configurations while the universal player takes care of
  universal configurations. The winning condition for the existential
  player is that a final configuration is eventually reached. In particular
  $\atm$ rejects its input \ssi the universal player has a winning
  strategy, \ie a playing policy ensuring no final state is
  ever reached whatever the existential player does. Moreover, this
  strategy can be chosen to be positional, \ie only depending on the
  configuration, not on how it was reached (see \eg
  \cite{GradelThomasWilke02}).
        
  Fix an input of the form $(a^n,\sigma \separator\sigma)\in\perm$. A
  strategy for the universal player and the proof that it is a winning one
  can be described as follows. For each position
  $0\leq j\leq 2n+2$ on the second tape (including the markers) we
  associate a $n+2$ tuple $\tau_j=(\tau_{0,j},\dots,\tau_{n+1,j})$ of
  functions describing what configurations can appear when the heads are at
  position $(i,j)$ and for universal ones what the strategy of the universal
  player is. Hence for each $0\leq i\leq n+1$, $\tau_{i,j}$ is a map from
  $Q$ such that:
  \begin{itemize}
  \item if $q\in Q_\exists$, $\tau_{i,j}(q)$ is either $\bot$ (meaning
    $(q,i,j)$ is not reachable) or $\top$ (reachable)
  \item if $q\in Q_\forall$, $\tau_{i,j}(q)$ is either $\bot$ (not
    reachable) or a transition in $\Delta$ starting with $(q,x,y)$ where
    $x$ (\resp $y$) is the letter at position $i$ (\resp $j$) in the first
    (\resp second) tape, including markers. The latter case means that
    configuration $(q,i,j)$ can appear and gives the corresponding move in
    the universal player strategy.
  \end{itemize}
  
  Such an object $\tau_0,\dots,\tau_{2n+2}$ is a \emph{proof of reject} if
  it satisfies the following four conditions.
  \begin{enumerate}
  \item For each initial state $q\in I$, $\tau_{0,0}(q)\neq \bot$ (all
    initial configurations are allowed).
  \item For each final state $q\in F$ and each $i,j$, $\tau_{i,j}(q)=\bot$ (no final
    configuration can be reached).
  \item For each existential state $q$ and each $i,j$ such that
    $\tau_{i,j}(q)=\top$, for each possible successor $(q',i',j')$ of
    $(q,i,j)$ one has that $\tau_{i',j'}(q')\neq\bot$ (all successors of
    allowed existential configurations are allowed).
  \item For each universal state $q$ and each $i,j$ such that
    $\tau_{i,j}(q)\neq\bot$, if one denotes by $(q',i',j')$ the
    configuration reached from $(q,i,j)$ applying transition
    $\tau_{i,j}(q)$ then $\tau_{i',j'}(q')\neq\bot$ (an allowed universal
    configurations has its successor described by the strategy allowed).
  \end{enumerate}
  
  It is then standard to notice that $\atm$ has no accepting run over
  $(a^n,\sigma \separator\sigma)$ \ssi there is a proof of reject for it.
  
  Now, for a fixed $n$, consider the number of possible values for the
  central part $(\tau_n,\tau_{n+1},\tau_{n+2})$ of a proof: it is smaller
  than $((|\Delta|+1)^{|Q|})^{3n}$. Hence, for $n$ large enough it is
  smaller than $n!$. For such $n$ it means that there are two distinct
  permutations $\sigma\neq \sigma'$ such that the proofs of reject for
  $\sigma \separator\sigma$ and for $\sigma' \separator\sigma'$ coincide on
  their central part: hence, they can be combined (glue the left part of
  the first proof with the right part of the second), leading a proof of
  reject for $\sigma \separator \sigma'$. But this leads a contradiction as
  $\sigma \separator\sigma'\notin \perm$ and therefore is accepted by
  $\atm$.
\qed\end{proof}

\section{Union, Intersection and Composition}

In this section, we study the closure under union, intersection and
composition of \models.

\subsection{Closure under Union and Intersection}

Concerning closure under union and intersection we have the following picture.

\begin{lemma}
  Relations recognized by deterministic (\resp non-deterministic, \resp
  alternating) \models are closed under union and intersection.
\end{lemma}

\begin{proof}[sketch]
  For non-deterministic automata union is for free; intersection is simply
  obtained by simulating the first automaton, and if it reaches an
  accepting state, simulating the second automaton. For alternating
  automata both closures are for free. For deterministic automata,
  intersection can be obtained as in the non-deterministic case. Concerning
  union, one can no longer use non-determinism and hence, one needs to
  simulate successively the two automata.  However this does not work in
  general as a \model can reject by entering an infinite loop. Hence, one
  must first get rid of such phenomenon: this can be achieved thanks to a
  method due to Sipser~\cite{Sipser80}, which ensures that a deterministic
  \model never rejects by looping.
\qed\end{proof}

Remark that preventing deterministic \models from rejecting by looping can
be used to prove that the class of relations they recognize is also closed under
complementation, and therefore forms a Boolean algebra.

\subsection{Non-closure under Composition}

In this section, we prove that the class of relations recognized by \models
is not closed under composition, even in the deterministic case.

\begin{theorem} \label{theo:composition}
  Relations recognized by deterministic (\resp non-deterministic, \resp
  alternating) \models are not closed under composition.
\end{theorem}

\begin{proof}
  The proof works the same way for all three models. One first establishes
  (see Lemma~\ref{boundFuncs} below) an elementary bound on the growth of
  the functions recognized by our model, and then exhibit (see
  Lemma~\ref{counting} below) a recognizable relation breaking this bound
  when self-composed.
\qed\end{proof}

We now give the lemmas used in the proof of Theorem~\ref{theo:composition}.
The following lemma bounds the growth of the functions, that is functional
relations recognized by \models (a similar statement can easily be obtained
for \amodels with a doubly exponential bound instead).

\begin{lemma}  \label{boundFuncs}
  If $\frel$ is a function recognized by a \model, then there
  exists $k \in \N$ such that for all $n \in \N$ and $ u \in \alphabet^n$,
  one has $\size{\frel(u)} \leq \binom{2nk}{nk + 1}$, where $\binom{n}{k}$
  denotes the binomial coefficient.
\end{lemma}

\begin{proof}
  Let $\atm$ be a \model recognizing a function $\frel$. Let $n
  \in \N$, and $u \in \alphabet^n$. Finally, let $k$ be the size of $\atm$.
  Fixing $u$, we easily build from $\atm$ a \omodel with $nk$ states
  accepting the singleton language $\{\frel(u)\}$.  Now, thanks to a result
  by Kapoutsis~\cite{Kapoutsis05} establishing that any $m$-state \omodel
  can be transformed into an equivalent equivalent \nfa with
  $\binom{2m}{m+1}$ states, we have that $\{\frel(u)\}$ is accepted by a
  \nfa with $\binom{2nk}{nk+1}$ states. As the shortest word accepted by an
  $m$-state \nfa has length at most $m$, it concludes the proof.
\qed\end{proof}

For $i \in \N$ and $n \geq \log_2(i)$ let $\base[n]{2}{i} \in \{0,1\}^n$ be
the writing of $i$ in base~$2$ on~$n$ bits, with the less significant bit
on the left (possibly padded with zeros). For example, $\base[5]{2}{6} =
01100$. We are now ready to exhibit a function with an exponential growth
and recognized by a \dmodel.

\begin{lemma}  \label{counting}
  The function $\enum = \{(w, \#w_0\#w_1\#\ldots\#w_{2^{\size{w}} - 1}\#)
  \mid w_i = \base[\size{w}]{2}{i}\}$ is recognizable by a \dmodel.
\end{lemma}

The proof is based on the concept of so-called synchronous relations. A
relation is \emph{synchronous} if it is recognized by a synchronous
transducer, which can be simulated by a \model whose two heads always move
simultaneously to the right and which accepts once the two inputs words are
entirely scanned. An example of such a relation is $\incrementS =
\{(\base[n]{2}{i},\base[n]{2}{i+1})\mid n\geq 1 \text{ and } 0\leq i\leq
2^n-1\}$ since an integer can deterministically be incremented starting
from its less significant digit.  Note that a synchronous transducer can
always be made deterministic as it is a classical automaton over the
product alphabet.

\begin{lemma} \label{syncrels}  
  Let $\srel$ be a synchronous relation. Then $\{(w, \# u\# v\#) \mid (u,
  v) \in \srel, \size{w} = \size{u} = \size{v} \}$ is recognized by a
  \dmodel.
\end{lemma}

The rough idea to prove Lemma~\ref{syncrels} is to use the first tape to
keep track of the position when alternatively reading the $u$ and $v$ part
of the word written on the second tape. 
        
Now, using Lemma~\ref{syncrels} with relation $\incrementS$ we conclude
that the relation $\increment = \{ (w,
\#\base[\size{w}]{2}{i}\#\base[\size{w}]{2}{i+1}\# \mid 0 \leq i \leq
2^{\size{w}} - 1 \}$ is recognizable by a \dmodel.

Finally, $\enum$ is recognized by iterating the methods for $\increment$ 
(this is made possible thanks to the presence of $\#$) and stopping once
the second word is only composed of 1s.

\bibliographystyle{splncs03}
\bibliography{main} 

\end{document}